%% file: main.tex
\documentclass[copyright,creativecommons]{eptcs}
\providecommand{\event}{QAPL 2015} 

\usepackage{booktabs}
\usepackage{hyperref}
\usepackage{listings}
\usepackage{amsmath}
\usepackage{graphics}
\usepackage{expdlist}
\usepackage[USenglish,british,american, english]{babel}
\usepackage{stmaryrd}
\usepackage{eurosym}
\usepackage{float}
\usepackage{algorithm,algpseudocode}
\usepackage{tikz}
\usepackage{nameref}
\usepackage{amsthm}
\usetikzlibrary{trees}

\input{macros}

\hypersetup{
  colorlinks=true,
  urlcolor=blue,
}

\newcommand{\taglia}[1]{}

\newcommand{\eval}[1]{\llbracket #1 \rrbracket}

\newcommand{\evalstrong}[1]{\llbracket\hspace{-0.06cm}[ #1 \rrbracket\hspace{-0.06cm}]}

\def\0{{\mathbf 0}}
\def\1{{\mathbf 1}}

\makeatletter
\newcommand{\xRightarrow}[2][]{\ext@arrow 0359\Rightarrowfill@{#1}{#2}}
\makeatother

\begin{document}

\title{Semiring-based Specification Approaches for\\ Quantitative Security\thanks{This work has been partially supported by  PRIN 2010XSEMLC ``Security Horizons'' and ARTEMIS J.U.~SESAMO.}}


\author{Fabio Martinelli 
\institute{IIT-CNR, Pisa, Italy}
\email{fabio.martinelli@iit.cnr.it} \and
Ilaria Matteucci
\institute{IIT-CNR, Pisa, Italy}
\email{ilaria.matteucci@iit.cnr.it}
\and
Francesco Santini
\institute{IIT-CNR, Pisa, Italy}
\email{francesco.santini@iit.cnr.it}
}
\def\titlerunning{Quantitative-security Specification}
\def\authorrunning{F. Martinelli, I. Matteucci \& F. Santini}

\maketitle

\input{0-Abstract.tex}
\input{1-Intro.tex}
\vspace{-0.5cm}
\input{2-BackgroundGPA.tex}

\input{3-qgndc-single.tex}
\input{4-ApproachCSP.tex}

\input{3-Related.tex}

\input{5-Conclusion.tex}

\bibliographystyle{eptcs}
\bibliography{bibliobis}
\end{document}

%% file: macros.tex
\usepackage{tikz}
\usetikzlibrary{arrows}
\usepackage{amssymb,amsmath}
\usepackage{xspace}

\usepackage{versions}
\excludeversion{notes}

\newtheorem{example}{Example}[section]

\makeatletter
   \def\marginparright{\@mparswitchfalse}
   \def\marginparoutside{\@mparswitchtrue}
\makeatother
\def\comment#1{#1}

\def\withcomments{
  \setlength{\marginparwidth}{0.3in}
  \addtolength{\oddsidemargin}{0in}
  \addtolength{\evensidemargin}{0in}
  \newcounter{mycommentcounter}
  \marginparright
  \def\comment##1{\refstepcounter{mycommentcounter}%
    \ifhmode%
    \unskip%
    {\dimen1=\baselineskip \divide\dimen1 by 2 %
      \raise\dimen1\llap{\tiny {\bf \color{red}*-\themycommentcounter-*}}}\fi%
    \marginpar{\renewcommand{\baselinestretch}{0.8}%
      \footnotesize [\themycommentcounter]: \raggedright ##1}}
  \date{\framebox{Draft of \today}}
}

\newcommand{\ie}{i.e.,\xspace}
\newcommand{\eg}{\emph{e.g.}, \xspace}

\newcommand{\cbot}{\bot}
\newcommand{\ctop}{\top}
\newcommand{\ctimes}{\ensuremath{\times}}
\newcommand{\cplus}{\ensuremath{+}}
\newcommand{\isequente}[2]{\protect{\frac{\displaystyle{#1}}{\displaystyle{#2}}}}
\newcommand{\arco}[1]{\stackrel{#1}{\rightarrow}}

\newcommand{\lang}{\mathcal{T}}
\newcommand{\weight}[1]{|#1|}
\newcommand{\lab}[1]{l(#1)}

\newtheorem{theorem}{Theorem}[section]
\newtheorem{definition}{Definition}[section]

\newtheorem{proposition}{Proposition}[section]

%% file: 0-Abstract.tex
\begin{abstract}
Our goal is to provide different semiring-based formal tools for the specification of security requirements: we quantitatively enhance the \emph{open-system} approach, according to which a system is partially specified. Therefore, we suppose the existence of an unknown and possibly malicious agent that interacts in parallel with the system. Two specification frameworks are designed along two different (but still related) lines. First, by comparing the behaviour of a system with the expected one, or by checking if such system satisfies some security requirements: we investigate a novel approximate behavioural-equivalence for comparing processes behaviour, thus extending the \emph{Generalised Non Deducibility on Composition} (GNDC) approach with scores. As a second result, we equip a modal logic with \emph{semiring} values with the purpose to have a weight related to the satisfaction of a formula that specifies some requested property. Finally, we generalise the classical partial model-checking function, and we name it as  \emph{quantitative partial model-checking} in such a way  to  point out the necessary and sufficient conditions that a system has to satisfy in order to be considered as secure, with respect to a fixed security/functionality threshold-value. 
\end{abstract}

%% file: 1-Intro.tex
\section{Introduction}\label{sec:introduction}
The considerable amount of trust and decentralisation, coming with today's software systems, demands for a rigorous security analysis. Unfortunately, security is frequently in conflict with the functionality and performance requirements  of a system, making $100$\% security an impossible or overly expensive goal to be accomplished. For instance, non-functional requirements add to the picture costs, execution times, and rates. Therefore, the relevant question is not whether a system is secure, but rather how much security it provides under such ``soft'' constraints. Instead of a plain yes/no answer, quantitative levels of security can express different degrees of protection, and allow a security expert to reason about the trade-off between security and conflicting requirements (\eg~on performance). Quantitative security analysis~\cite{dagstuhl} has been already applied, \eg~to name a few, for quantifying the side-channel leakage in cryptographic algorithms, for capturing the loss of privacy in statistical data analysis or information flows, and for quantifying security in anonymity networks.

Improving a quantitative security-analysis requires different tools for the rigorous development of practical systems, and an extended formal foundation for the management of security risks. Here we focus on the latter task. 
The goal of this paper is to move from a qualitative interpretation of security to a quantitative one.
The basic ingredients in our ``recipe'' are c-semirings~\cite{jacm97,gadducci06} (or simply ``semirings'' in the following) and the \emph{Generalised Process Algebra} (\emph{GPA})~\cite{Buchholz}, a quantitative process-algebra where actions are labelled with a value taken from a semiring. Therefore, we use GPA to model processes with quantitative aspects: different semiring instantiations  can parametrically model different cost-metrics. In order to formalise security-properties of GPA processes, we provide two different approaches.

The {\bf first approach} consists in providing several definitions of quantitative behavioural-equivalencies in such a way to extend with quantities the family of security properties that can be expressed in \emph{Generalised Non Deducibility on Composition}  (\emph{GNDC})~\cite{FocardiM99}. The GNDC schema is a uniform approach for defining security properties derived from the \emph{Non Deducibility on Composition} (NDC) properties~\cite{NI,683116}. 
The GNDC scheme  uniformly  expresses many security properties as, \eg fault tolerance properties (\emph{fail stop, fail silent, fail safe} and \emph{fault tolerant} behaviour, \eg~\cite{GLMM09}) or, also, many security properties of cryptographic protocols as, \eg \emph{secrecy}, \emph{authentication}, \emph{integrity}, etc.~\cite{DBLP:conf/fosad/FocardiGM02}. Hence, we formalise the system through quantitative observational relations. We introduce the notion of \emph{quantitative trace-equivalence}, and we recall the definition of \emph{quantitative bisimulation} given in \cite{MiculanP13}. Furthermore, we extend both these relations by considering an approximate version of them: the $\epsilon$-equivalence. By using these equivalence relations, we can compare and specify different security properties, as a quantitative extension of NDC and bisimulation-based NDC properties (BNDC)~\cite{NI,683116}.

In the {\bf second approach} we present in this paper,  we first introduce a semiring-based extension of the classical \emph{Hennessy-Milner Logic} (named \emph{c-HM} Logic) as a means to quantitatively measure the satisfaction of a given formula: its truth value can now be not only true/false, but a numeric value as well (\eg $50$\% or $3$\euro). Note that by exploiting the boolean semiring (\ie~$\langle \{\mathit{false},\mathit{true}\}, \vee, \wedge, \mathit{false}, \mathit{true}\rangle$) we can still enforce yes/no only requirements.
Hence, we use c-HM Logic in the frame of \emph{Partial Model Checking} (\emph{PMC})~\cite{Andersen}. Classical \emph{Model Checking} (\emph{MC}) involves using verification tools to exhaustively search in a process/protocol specification for all the execution sequences with some desired properties. PMC focuses this verification on part of a system only: the main advantage is to perform a full  analysis while avoiding the combinatorial explosion of the state space.
In security, the PMC function has been often used to point out necessary and sufficient constraints on the unspecified/unknown part of a system that is supposed to show a malicious behaviour. Hence,  a controller program is required to ensure the correct behaviour of the whole system, comprehensive of the attacker~\cite{DBLP:journals/entcs/MartinelliM07}. 
In a quantitative scenario, we associate the notion of satisfiability of a logic formula  with the security/functionality level of a system. Once 
we set a \emph{satisfiability threshold} $t \in K$, if the system quantitatively satisfies a security requirement $\phi$ with a value $k$ worse than $t$, then we can state that the investigated system is not quantitatively secure. 


The paper is structured as follows. In Sec.~\ref{sec:background} we recall c-semiring algebraic structures and GPAs. In Sec.~\ref{sec:qgndc} we introduce our first approach, which aims at comparing  a system behaviour with the expected one: we adopt both trace and bisimulation equivalence. Hence, we rephrase them as approximate relations, in order to include  ``close''-enough processes, where close is related to a threshold-score $\epsilon$. In this way, we are able to specify some security aspects formalised as a quantitative  GNDC schema. In Sec.~\ref{sec:logic} we describe security properties via  a semiring-based modal logic (\ie c-HM), and in Sec.~\ref{sec:qpmc} we define a QPMC function with the purpose to point out the necessary and sufficient conditions each subsystem has to satisfy for guaranteeing such requirements. 
Finally, Sec.~\ref{sec:related}  summarises the related work in literature, and Sec.~\ref{sec:conclusion} wraps up the paper with conclusions and proposes some future work.

%% file: 2-BackgroundGPA.tex
\section{Background}\label{sec:background}
In this section we recall the necessary fundamental notions about c-semirings~\cite{jacm97,gadducci06} and \emph{Generalised Process Algebra}~\cite{Buchholz}, a quantitative process-algebra based on semirings.
\subsection{Semirings}\label{sec:semirings}
\begin{definition}[semiring~\cite{golan}]
A commutative semiring is a five-tuple $\mathbb K= \langle K,\cplus,\ctimes,\cbot,\ctop \rangle$  such that $K$ is a set, $\ctop, \cbot \in K$, and $\cplus, \times : K \times K \rightarrow K$
are binary operators making the triples $\langle K, \cplus, \cbot \rangle$ and $\langle K, \ctimes, \ctop \rangle$
commutative monoids (semigroups with identity), satisfying
\begin{itemize}
\item (distributivity) $\forall a, b, c \in K.a \ctimes (b + c) = (a \ctimes b) + (a \ctimes c)$.
\item (annihilator) $\forall a \in A.a \ctimes \cbot = \cbot$.
\end{itemize}
\end{definition}
\begin{proposition}[absorptive semirings~{\cite{golan}}]
Let $\mathbb K$ be a commutative semiring. Then these two properties are equivalent:
\begin{itemize}
\item (absorptiveness) $\forall a,b \in K. a\cplus(a\ctimes b) = a$.
\item ($\ctop$ absorbing element of $+$) $\forall a \in K. a \cplus \ctop = \ctop$.
\end{itemize}
\end{proposition}

Absorptive semirings are referred also as \emph{simple}, and their $+$ operator is necessarily idempotent~\cite[Ch.~1,\;pp.~14]{golan}. Semirings where $+$ is idempotent are called as \emph{dioids}.

\begin{definition}[c-semiring~\cite{jacm97}] 
C-semirings are commutative and absorptive semirings. Therefore, c-semirings are dioids where $\ctop$ is an absorbing element for $+$.
\end{definition}
The idempotency of $\cplus$ leads to the definition of  a partial ordering $\leq_K$ over the set $K$ ($K$ is a poset). Such partial order is defined as $a \leq_K b$ if and only if $a+b = b$, and $\cplus$ becomes the \emph{least upper bound} ($\mathit{lub}$, or $\sqcup$) of the lattice $\langle K, \leq_K\rangle$. This  intuitively means that $b$ is ``better'' than $a$.  As a consequence, we can use $\cplus$ as an optimisation operator and always choose the best available solution.

Some more properties can be derived on c-semirings~\cite{jacm97}: \emph{i)} both $\cplus$ and $\ctimes$ are monotone over $\leq_K$, \emph{ii)} $\ctimes$ is intensive (\ie~$a \ctimes b \leq_K a$), iii) $\ctimes$ is closed (\ie~$a \ctimes b \in K$), and  \emph{iv)} $\langle K, \leq_K\rangle$ is a complete lattice. $\cbot$ and $\ctop$ are respectively the bottom and top elements of such lattice.  When also $\ctimes$ is idempotent, \emph{i)} $\cplus$ distributes over $\ctimes$, \emph{ii)} $\ctimes$ is the \emph{greater lower bound} ($\mathit{glb}$, or $\sqcap$) of the lattice, and \emph{iii)} $\langle K, \leq_K\rangle$ is a distributive lattice.

Semirings and c-semirings have been  often adopted  in Computer Science and Operation Research as a very simple but very expressive optimisation structure~\cite{csop}. Some c-semiring instances are:
%
\emph{boolean} $\langle \{\mathit{F},\mathit{T}\}, \vee,$ $\wedge, \mathit{F}, \mathit{T}\rangle$\footnote{Boolean c-semirings can be used to model crisp problems.},
\emph{fuzzy} $\langle [0,1],$ $ \max, \min, 0, 1 \rangle$, 
\emph{bottleneck}  $\langle \mathbb{R}^+ \cup\{+\infty\}, $ $ \max, \min, 0, \infty \rangle$,
\emph{probabilistic} $\langle [0,1], \max, \hat{\times}, 0, 1 \rangle$ (known as the Viterbi semiring), 
\emph{tropical} $\langle \mathbb{N} \cup\{+\infty\}, min, \hat{+}, +\infty, 0 \rangle$. Capped operators stand for their arithmetic equivalent.

Although c-semirings  have been historically
used as monotonic structures where to aggregate costs (and find best solutions), the need of removing values has raised in local consistency algorithms and non-monotonic algebras using constraints (eg~\cite{gadducci06}).
A solution comes from \emph{residuation theory}~\cite{residuation},
a standard tool on tropical arithmetics  that allows for obtaining a division operator
via an approximate solution to the equation $b \ctimes x = a$.

\begin{definition}[division~\cite{gadducci06}]
Let $\mathbb K$ be a tropical semiring. Then, $\mathbb K$ is residuated if
the set $\{x \in K \mid b \ctimes x \leq a\}$ admits a maximum for all elements
$a, b \in K$, denoted as $a \div b$.
\end{definition}

Since a complete\footnote{$\mathbb K$ is complete if it
is closed with respect to infinite sums, and the distributivity law holds
also for an infinite number of summands~\cite{gadducci06}.} dioid is also residuated, we have that all the classical instances of a c-semiring presented above are residuated, \ie each element in $K$ admits an ``inverse'', which is unique in case $\leq_K$ is a total order. For instance, the unique ``inverse'' $a \div b$ in the weighted semiring is defined as follows:\\
$ a \div b
= \min \{x \mid b \hat{+} x \geq a\}=
\begin{cases}
0 & \text{if $b\geq a$}\\
a \hat{-} b & \text{if $a > b$}
\end{cases}
$
\begin{definition}[unique invertibility~\cite{gadducci06}] Let $\mathbb K$ be an absorptive, invertible semiring. Then, $\mathbb K$
is uniquely invertible iff it is cancellative, \ie $\forall a, b, c \in A.(a \times c = b \times c) \wedge (c \not= 0) \Rightarrow a = b$.
\end{definition}

Note that since all the previously listed semirings (\eg tropical and fuzzy) are cancellative, they are uniquely invertible as well.
Furthermore, it is also possible to consider several optimisation criteria at the same time: the cartesian product of semirings is still a semiring. Clearly, in this case the ordering induced by $+$ is partial, \eg when  we have $\langle k_1, k_2\rangle$ and $\langle k_3, k_4 \rangle$, and $k_1 \leq k_3$ while $k_2 \geq k_4$.

%

\subsection{Generalised Process Algebra}\label{sec:GPA}
In a \emph{quantitative process}, observable transitions are labelled with some value associated with a step in the behaviour of a system. In GPA~\cite{Buchholz} the authors use semirings to model two fundamental modes of composing observable behaviour, either by combination of different traces, or by sequential composition. 
%
Process algebras are simple languages with precise mathematical semantics, tailored to exhibit and study specific features of computation. Typically, a \emph{process} $P$, specified by some syntax, may non-deterministically execute several \emph{labelled transitions} of the form $P \arco{a} P'$, where $a$ is an observable effect and $P'$ is a new process. In quantitative process algebras, transitions are labelled by pairs $(a,k)$ where $k$ is a quantity associated to the effect $a$: thus, $P \xrightarrow{(a,k)} P'$. 

We define transition systems where transitions are labelled with
symbols from a finite alphabet and from a semiring $\mathbb K$. The semantics of a GPA process $P$ is \emph{Multi Labelled Transition System} (\emph{MLTS})~\cite{Buchholz}:

\begin{definition}[MLTS]\label{def:MLTS} A (finite) Multi Labelled Transition System (MLTS) is a five-tuple
$\mathit{MLTS} = (S, \mathit{Act}, \mathbb{K}, T, s_0)$, where $S$ is the countable (finite) state space, $s_0 \in S$ is the initial state,\footnote{We simplify the original  definition of MLTS given in~\cite{Buchholz}, where an \emph{initialization} function is  taken into account to assign a quantitative valuation to each of the $n$ initial states (here we only have one $s_0$).}
$\mathit{Act}$ is a finite set of transition labels, $\mathbb{K}$ is a semiring used for the definition
of transition costs, and $T : (S \times \mathit{Act}  \times S) \longrightarrow \mathbb{K}$ is the transition function.
\end{definition}
%
%
\begin{definition}[GPA syntax~\cite{Buchholz}]
The set $\mathcal{P}$ of \emph{agents}, or processes, in GPA over a countable set of transition labels Act and a semiring $\mathbb{K}$ is defined by the grammar
\small$$
P ::= 0 \mid (a,k).P \mid P + P \mid P \Vert_{A}\, P \mid P\backslash A \mid P \slash A
\mid X \mid X \triangleq P
$$
where $a \in Act$, $A\subseteq Act\backslash \{\tau\}$ is a subset of actions, $k \in K$ (the set of values in a semiring $\mathbb{K}$), and $X$ belongs to a countable set of \emph{process variables}, coming from a system of co-recursive equations of the form $X \triangleq P$, meaning that $X$ behaves like $P$. 
$GPA(\mathbb{K})$ denotes the set of GPA processes labelled with weights in $\mathbb{K}$. 
\end{definition}

The formal operational semantics of GPA operators is given in Tab.~\ref{GPASOS}, Informally, process $0$ describes inaction or termination;
$(a,k).P$ performs $a$ with \emph{value} $k$ and evolves into $P$; $P + P'$ non deterministically behaves as either $P$ or $P'$;
$P \Vert_{A}\, P'$ describes the process in which $P$ and $P'$ proceed concurrently  when they perform actions belonging to $A$, and independently on all the other actions; $P \backslash A$ expresses the fact that actions from the set $A$ are hidden, \ie they become $\tau$ actions that
are no longer usable in joint actions with an environment; the dual, \ie $P \slash A$, restricts the behaviour of P by allowing it to perform only actions not in $A$.
\begin{table}[t] \hrulefill
\begin{center}
{\small
  \begin{tabular}{ccc}
     $ \isequente{}{(a,k).P\arco{a,k} P}$ &       $ \isequente{P\xrightarrow{(a,k)} P_{1}\quad P' \xrightarrow{(a,l)}P'_{1}} {P\Vert_{A}\, P' \xrightarrow{(a,k \ctimes l)}P_{1}\Vert_{A}\, P'_{1}} \; a \in A$  &
 $\isequente{P\xrightarrow{(a,k)} P_{1}}{X\xrightarrow{(a,k)} P_{1}}{X\triangleq P}$ \\ \\
   $ \isequente{P\xrightarrow{(a,k)} P_{1}}{P \Vert_A, P'\xrightarrow{(a,k)} P_{1}\Vert_A\, P'}\; a \not \in A$ &
      $\isequente{P_{j}\xrightarrow{(a,k)} P_{1}}{\sum_{i\in I} P_{i}\xrightarrow{(a, k_{\Sigma})} P_{1}}{j \in I \;}$  &
     $ \isequente{P'\xrightarrow{(a,k)}P'_{1}}{P \Vert_A\, P'\xrightarrow{(a,k)} P\Vert_A\, P'_{1}} \; a \not \in A$ \\  
     $\isequente{P'\xrightarrow{(a,k)}P'_{1}}{P \backslash A\, \xrightarrow{(a,k)} \, P'_{1}\backslash A} \; a \not \in A$ & $\isequente{P\xrightarrow{(a_{1},k_{1})}P' \ldots P\xrightarrow{(a_{n},k_{n})}P' }{P \backslash A\, \xrightarrow{(\tau,k_{\tau})} \, P'\backslash A } \; \{a_{1}, \ldots a_{n}\}  \subseteq A \cup \{\tau\}$ & $\isequente{P'\xrightarrow{(a,k)}P'_{1}}{P \slash A\, \xrightarrow{(a,k)} \, P'_{1}\slash A } \; a \not \in A$
  \end{tabular}}
 \end{center}
  \hrulefill
  \caption{An operational semantics for $GPA$~\cite{Buchholz}, where $k_{\Sigma} = \sum_{i\in I}(P_{i} \arco{a}P_{1})$ and $k_{\tau}= \sum_{i=1}^{n}(k_{i})$.}\label{GPASOS}
\end{table}

Given a GPA process $P$, the set of \emph{derivatives} of a $P$ is defined as $Der(P)=\{P' \mid P \rightarrow^* P'\}$ where $\rightarrow^{*}$ is $\cup_{a\in Act, k \in K} \arco{a,k}$; $Sort(P)$ denotes the set of actions names that syntactically appear in $P$ regardless their values.

Being $a_{1},\ldots,  a_{n} \in Act$, a \emph{trace} is a sequence $(a_1,k_1)$ $\cdots(a_n,k_n)$ leading from process $P$ to process $Q$. We call $\lang(P)$ the set of traces rooted in $P$.
Given a trace $(a_1,k_1)$ $\cdots (a_n,k_n)$, we define its \emph{label} $\lab{t}=a_1 \cdots a_n$, and its \emph{weak run-weight} $\weight{t} =  k_1 \ctimes \ldots \ctimes k_n \in K$ (where $\times$ comes from a semiring $\mathbb{K}$). 
We also define the \emph{strong run-weight} $\Vert{t}\Vert$ of a trace,  as the weak-run weight without the weights of $\tau$ actions.

Hence, it is possible to \emph{evaluate} the whole behaviour of a process. The valuation of the $0$ process is equal to $\ctop$. We consider processes different form $0$ as evaluated in the \emph{optimistic} way, \ie their evaluation  coincides with the value of their best trace(s). Formally, given a process $P \neq 0$, the \emph{weak evaluation-value} is computed as $$\eval{P} = \sum\limits^{\mathbb K}_{\{t \in \lang(P)\}} \weight{t},$$ where $\sum\limits^{\mathbb K}$ is the set-wise version of the $+$ operator in $\mathbb K$. The \emph{strong evaluation-value} is  computed as $$\evalstrong{P} = \sum\limits^{\mathbb K}_{\{t \in \lang(P)\}} \Vert t \Vert.$$ 



%% file: 3-qgndc-single.tex
\section{Quantitative Generalized Non Deducibility on Composition}\label{sec:qgndc}

The GNDC schema is a uniform approach for defining several security properties based on the compositionality nature of the process algebra formalism. It has been introduced in~\cite{FocardiM99} to express security properties in a qualitative way. Hereafter, we extend that definition in order to express, in a uniform way, quantitative security properties. Therefore, what we achieve is to be able to quantitatively compare the behaviour of two GPA processes, according to possible different definitions of quantitative behavioural relations (\eg a weighted trace-equivalence relation). 

Hence, we have  a quantitative version of the GNDC schema, hereafter denoted as QGNDC, given in terms of GPA:
\begin{equation} \label{start}
P \in QGNDC_{\triangleleft}^{\alpha, \mathbb{K}} \; \mbox{iff}\; \forall E
\in \mathcal{E}_{H} : (P \Vert_{H} X) \backslash H
\triangleleft_{\mathbb{K}} \alpha(P)
\end{equation}
where $H \subseteq Act\backslash \{\tau\}$ is the set of environmental actions,
$\mathcal{E}_{H}$ is the set of environments, $\triangleleft_{\mathbb{K}} \in \mathcal{P} \times \mathcal{P}$
is a relation between two processes, whose definition depends on the partial order of the semiring $\mathbb{K}$ according to which the processes are quantified and evaluated, and $\alpha: \mathcal{P}
\rightarrow \mathcal{P}$ is a function between processes. The $\Vert_{H}$ is the
synchronisation operator stating that all actions in $H$ are
performed by the system if and only if both $P$ and $E$
perform them, and the $\backslash H$ is the hiding operator that hides all actions in $H$. \\
Informally, the $GNDC_{\triangleleft}^{\alpha, \mathbb{K}}$ property requires that the behaviour of  process $P$, once it is composed with any possible environment $E \in \mathcal{E}_{H}$, is \emph{compliant} with the system expected-behaviour, described by the function $\alpha$. 
The notion of compliance depends on the $\triangleleft_{\mathbb{K}}$ relation we select for comparing the behaviours of $(P \Vert_{H} X)\backslash H$ and $\alpha(P)$, according not only to an observational equivalence (as in the qualitative approach~\cite{FocardiM99}), but also with respect to order induced by the semiring $\mathbb{K}$.

In the following  we provide several definitions of quantitative behavioural-equivalence according to which we are able to specify weighted  properties through the QGNDC schema~\cite{FocardiM99}.
Furthermore, we compare the expressive power of the different equivalence-relations we define.

\subsection{Quantitative Trace-equivalences}\label{sec:tracestuff}
One of the basic notions used in the literature to compare processes behaviour is the notion of \emph{trace}: 
two processes are equivalent if they exactly show the same execution sequences, ands their evaluation scores are comparable in the semiring partial-order.
In order to formally define traces, we need a transition relation that does not consider internal moves, denoted by $\tau$. We start by highlighting such $\tau$-actions in execution traces:
%

\begin{definition}[weighted weak-trace]
The notation $P\xRightarrow{(a,k)}P'$ is a shorthand for $P\xrightarrow{\smash{(\tau,k_{\tau})}}^* P_{\tau} \xrightarrow{\smash{(a,k)}}P'_{\tau} \xrightarrow{\smash{(\tau,k'_{\tau})}}^*P'$, where a (possibly empty) sequence of $\tau$ labeled transitions is denoted by $\xrightarrow{\smash{(\tau,k_{\tau})}}^*$. 
A \emph{weighted weak-trace}   $\gamma= (a_1,k_1) \ldots(a_n,k_n) \in (Act \backslash \{\tau\})^*$ is such that $P \xRightarrow{\; \; \; \gamma \; \; \;} P'$ if and only if  there exist $P_{1}, \ldots, P_{n-1} \in GPA$ such that $P \xRightarrow{(a_1,k_1)}P_{1} \ldots P_{n-1} \xRightarrow{(a_n,k_n)}P'$.
\end{definition}

We can now define an equivalence relation based on trace similarity, \ie  the \emph{weak-trace equivalence} ($\approx_{wtrace}$) in Def.~\ref{def:tracerel}. We require both the strong evaluation-score and the weak evaluation-score of two processes to be equal, or not comparable:

\begin{definition}[weak-trace equivalence]\label{def:tracerel}
For any $P \in {\cal P}$ the set $\hat{\mathcal{T}}(P)$ of {\em weighted weak-traces associated with \/} $P$ is $\hat{\mathcal{T}}(P) = \{\gamma \in (Act \backslash \{\tau\})^*\; \mid\; \exists
P'\;:\; P \xRightarrow{\; \; \gamma \; \;} P'\}$, where $(Act \backslash \{\tau\})^*$ is the set of sequences of actions.
$P$ and $Q$ are {\em weak-trace equivalent\/}
(notation $P \approx_{wtrace} Q$) if and only if 
all the following three conditions hold:
\begin{enumerate}
\item $\hat{\mathcal{T}}(P) = \hat{\mathcal{T}}(Q)$,
\item $\evalstrong{P} \not\lessgtr_\mathbb{K}  \evalstrong{Q}$,\footnote{In the following we will use $\not\lessgtr_\mathbb{K}$ as a shortcut to denote when 
two semiring values are equal or  not comparable in the poset.}
 and 
\item $\eval{P} \not\lessgtr_\mathbb{K}  \eval{Q}$.\end{enumerate} 
\end{definition}
%


Note that, the first two conditions are related to the observable traces of $P$ and $Q$, while condition 3 allows us to compare the specific contribution of the $\tau$-actions in terms of weight. In the following, we provide an approximate version of weak-trace equivalence, \ie the \emph{$\epsilon$-trace relation}. With respect to Def.~\ref{def:tracerel}, we allow the weak evaluation-score of two processes to differ up to a threshold-value $\epsilon \in K$.  

\begin{definition}[$\epsilon$-trace equivalence]\label{def:epsilon-trace}
For any $P \in {\cal P}$ the set $\hat{\mathcal{T}}(P)$ of {\em weighted weak-traces associated with \/} $P$ is $\hat{\mathcal{T}}(P) = \{\gamma \in (Act \backslash \{\tau\})^*\; \mid\; \exists
P'\;:\; P \xRightarrow{\; \; \gamma \; \;} P'\}$, where $(Act \backslash \{\tau\})^*$ is the set of sequences of actions.
$P$ and $Q$ are {\em $\epsilon$-trace equivalent\/}
(notation $P \approx_{\epsilon-trace} Q$) if and only if there exists a value $\epsilon$ such that 
all the following three conditions hold:
\begin{enumerate}
\item $\hat{\mathcal{T}}(P) = \hat{\mathcal{T}}(Q)$,
\item $\evalstrong{P} \not\lessgtr_\mathbb{K}  \eval{Q}$, and 
\item $\eval{P} \div \epsilon \geq_\mathbb{K} \eval{Q}  \wedge \eval{Q}  \div \epsilon \geq_\mathbb{K} \eval{P} $.
\end{enumerate} 

\end{definition}
%

%
%
%
%
%
%
%
%
%
%
%
%
%
%

These relations are comparable one to another. In particular, the following proposition holds.
\begin{proposition}\label{prop:trace}
For each couple of processes $P, Q \in GPA$. The following statement holds
$$
\forall \epsilon \in K, \quad P \approx_{wtrace} Q \; \Rightarrow  P\approx_{\epsilon-trace} Q
$$
Note that when $\epsilon = \ctop$ we have $P \approx_{wtrace} Q \; \Leftrightarrow  P\approx_{\epsilon-trace} Q$.
\end{proposition}
%

\begin{example}\label{ex:gino}
Consider two processes $P = (\tau,1).(a,3).(b,2)$ and $Q= (a,2).(b,3)$ in the tropical semiring. We have that $P  \approx_{1-trace} Q$  (\ie $\epsilon = 1$), while $P \approx_{wtrace} Q$ does not hold.
\end{example}

Note that $P$ and $Q$ in Ex.~\ref{ex:gino} are qualitatively trace-equivalent according to the classic definition given in \cite{FocardiM99}. Therefore, by considering the weight of traces (\ie weak-trace equivalence) we obtain a more restrictive equivalence-relation. Consequently, we have introduced the notion $\epsilon$-trace equivalence with the purpose to gradually be able to relax it and include more processes in the relation. 

%
%

\subsection{Quantitative Bisimulation Equivalences}\label{sec:qbe}
%

%
%

In this section we focus on the weak-bisimulation equivalence for GPA~\cite{Buchholz,MiculanP13}, since we would like to consider as equivalent the behaviour of two processes regardless the weight of internal action $\tau$ they perform. Differently from~\cite{Buchholz}, where only the definition of strong bisimulation is provided, we assume that each state of a MLTS has a finite number of transitions with a non-$\ctop$ weight. 
In the following, for $\mathcal R$ a relation, we write $P \mathcal R Q$ to say that $(P, Q) \in \mathcal R$. 

We extend the definition of quantitative weak bisimulation in \cite{MiculanP13} by considering a poset of preference values:
\begin{definition}[quantitative weak-bisimulation]\label{def:weak-bis}
An equivalence relation $\mathcal{R}$ on $\mathcal{P} \times\mathcal{P}$ is a \emph{quantitative weak bisimulation} if and only if for all $(P, Q) \in \mathcal{R}$ and all $a \in Act$ and each equivalence class $C \in \mathcal{R}$ we have:
$$
\sum_{D \in C} (P \xRightarrow{(a, k)}D) \not\lessgtr  \sum_{D \in C} (Q \xRightarrow{(a, k')}D), \quad \quad
\sum_{D \in C} (P \xrightarrow{\smash{(\tau, k_\tau)}}^*D) \not\lessgtr \sum_{D \in C} (Q \xrightarrow{\smash{(\tau, k'_\tau)}}^{*}D)
$$
\normalsize
We write $P \approx_\mathbb{K} Q$ whenever there is a bisimulation $\mathcal{R}$ such that $(P, Q) \in \mathcal{R}$.
\end{definition}
Note that the quantitative weak-bisimulation relation holds even if the two values related to $P$ and $Q$ are incomparable in the partial order defined by $+$. In \cite{MiculanP13} they have to exactly correspond to the same value, since partial orders are not considered.

As accomplished in Sec.~\ref{sec:tracestuff},  we define a variant that approximates Def.~\ref{def:weak-bis}, named as \emph{weak $\epsilon$-bisimulation}. 
The intuition behind it, similarly to Sec.~\ref{sec:tracestuff}, is to relax the cost of $\tau$ actions by a threshold-value $\epsilon$ with the purpose to allow two processes to be bismilar (or, better, $\epsilon$-bisimilar) despite this difference. More precisely, such $\epsilon$ value bounds the difference between the cost of $\tau$ actions before and after an action at the same time (see Ex.~\ref{ex:gino2}).

\begin{definition}[weak $\epsilon$-bisimulation]\label{def:epsilon}
An equivalence relation $\mathcal{R}$ on $\mathcal{P} \times\mathcal{P}$ is a \emph{weak $\epsilon$-bisimulation} if and only if, there exists a value $\epsilon$ such that for all $(P, Q) \in \mathcal{R}$ and all $a \in Act$ and each equivalence class $C \in \mathcal{R}$ we have:
\small{$$
\begin{array}{lclcl}
\sum\limits_{D \in C} (P \xRightarrow{(a, k)}D)  \div \epsilon&\geq_\mathbb{K}& \sum\limits_{D \in C} (Q \xrightarrow{(a, k')}D) \; \; \; \wedge \; \; \; \sum\limits_{D \in C} (Q \xRightarrow{(a, k)}D)  \div \epsilon&\geq_\mathbb{K}& \sum\limits_{D \in C} (P \xrightarrow{(a, k')}D) 
\end{array}
$$
}
\small{
$$
\begin{array}{lclcl}
\sum\limits_{D \in C} (P \xrightarrow{\smash{\tau, k_\tau}}^*D) \div \epsilon &\geq_\mathbb{K} & \sum\limits_{D \in C} (Q \xrightarrow{\smash{\tau, k'_\tau}}^*D) \; \; \; \wedge \; \; \; \sum\limits_{D \in C} (Q \xrightarrow{\smash{\tau, k_\tau}}^*D) \div \epsilon &\geq_\mathbb{K}& \sum\limits_{D \in C} (P \xrightarrow{\smash{\tau, k'_\tau}}^*D) 
\end{array}
$$
}
\normalsize
We write $P \approx_{\epsilon} Q$ whenever there is a bisimulation $\mathcal{R}$ such that $(P, Q) \in \mathcal{R}$.
\end{definition}

%

%
%
%
%
%
%

These relations are comparable as follows.
\begin{proposition}\label{prop:bis}
For each couple of processes $P, Q \in GPA$. The following statement holds
$$
\forall \epsilon \in K \quad P \approx_\mathbb{K} Q \; \Rightarrow  P\approx_{\epsilon} Q
$$
Note that when $\epsilon = \ctop$ we have $P \approx_{\mathbb{K}} Q \; \Leftrightarrow  P\approx_{\epsilon} Q$
\end{proposition}

\begin{example}\label{ex:gino2}
Consider two processes $P = (\tau,3).(a,4).(\tau,5)$ and $Q= (\tau,2).(a, 4).(\tau, 1)(\tau,1)$ in the tropical semiring. We have that $P  \approx_{1} Q$  (\ie $\epsilon = 1$) while $P \approx_{\mathbb{K}} Q$ does not hold.
Instead, if we have two processes $W = (\tau,3).(a,4).(\tau,3)$ and $Y= (\tau,2).(a, 4).(\tau, 1).(\tau,1)$, $W  \approx_{2} Y$  (\ie $\epsilon = 2$) while $W  \approx_{1} Y$ does not hold.
\end{example}

Note that both $P$ and $Q$, and $W$ and $Y$ in Ex.~\ref{ex:gino2} are weak bisimilar according to the classic definition given in \cite{milner}. Therefore, by considering the bisimulation relation in Def.~\ref{def:weak-bis} we obtain a more restrictive equivalence-relation. 

%% file: 4-ApproachCSP.tex
\section{C-semiring H-M Logic}\label{sec:logic}
In the previous section, we have shown how quantitative security properties  can be specified by using several quantitative process-equivalences in order to compare the behaviour of a system with respect to the expected one. 
A different approach for specifying  quantitative security-requirements is to express them as a logic formula that the system has to satisfy. It can be useful, for instance, when it is not decidable if two processes are quantitatively equivalent (as defined in Sec.~\ref{sec:qgndc}). Furthermore, 
%
%
%
%
some properties as, for example, \emph{safety properties}\footnote{E.g., properties expressing that, if something goes wrong, it can be detected in a finite number of steps}, can be easily expressed through a logic formula and allow for not requiring the behaviour of the whole system to be checked~\cite{683116,DBLP:journals/entcs/MartinelliM07}.

For this reason, in the rest of this section we propose a different approach with respect to the one described in Sec.~\ref{sec:qgndc}, with the purpose to advance an alternative methodology to quantitatively specify the security of a system. Such approach is based on Model Checking and a satisfiability procedure, instead of behavioural equivalences and a comparison checking. 

Hence, in order to specify whether a system is secure or not, we need to require that it satisfies a logic formula expressing the intended security-requirements.
To this aim, next we propose a quantitative variant of the Hennessy-Milner logic, named c-HM, in such a way to be able to specify a quantitative formula. In particular, differently from~\cite{LlM05}, where weights are associated to system states, in our approach values are part of  transition labels (together with an action): again we consider a MLTS (see Def.~\ref{def:MLTS}), and we evaluate the satisfaction of a c-HM formula over processes expressed in GPA.
In Def.~\ref{def:syntax}, we syntactically define the set $\Phi_M$ of correct formulas given  an MLTS $M$.

\begin{definition}[syntax]\label{def:syntax}
Given a MLTS $M = \langle S,$ $\mathit{Act}, \mathbb K, T, s_0\rangle$, and let $a \in \mathit{Act}$, a formula $\phi \in \Phi_M$ is syntactically expressed as follows, where $k \in K$:
$$
\phi ::= k \mid  
\phi_1 + \phi_2 \mid \phi_1 \times \phi_2 \mid \phi_1 \sqcap \phi_2 \mid \langle a \rangle \phi \mid [a] \phi
$$
%
\end{definition}
Clearly we can express more than just true (corresponding to $\ctop \in K$) and false ($\cbot \in K$) through all the values  $k\in K$. 
Semiring operators $+$ (the lub $\sqcup$), glb $\sqcap$, and $\times$ are used in place of classical logic operators $\vee$ and $\wedge$, in order to compose the truth values of two formulas together. As a reminder, when the $\times$ operator is idempotent, then $\times$ and $\sqcap$ coincide (see Sec.~\ref{sec:background}).
Finally, we have the two classical modal operators, \ie~``possibly'' ($\langle\cdot\rangle$), and ``necessarily'' ($[\cdot]$).


It is also possible to have a negation operator $\neg : K \longrightarrow K$, which is a unary operator such that, being $A \subseteq Act$, 
$\neg a \in A$ and $\neg\neg(a) = a$ for all $a \in A$, and $\neg \bigsqcup \{A'\} = \{\neg a \mid a \in A\}$ for all 
$A' \subseteq A$, where $\bigsqcup$ and $\bigsqcap$ are the set-wise lub and glb operators of the lattice $\langle A, \leq_K \rangle$. The negation operator allows us to use the equivalence $\neg \cbot = \ctop$. Note that 
the duality $\neg(a+b) = (\neg a) \times (\neg b)$ holds exactly when $\times$ is idempotent. Some examples where  
negation can be defined are the logical c-semiring, where logical negation is a negation operator, and 
probabilistic and fuzzy c-semirings, where $1-$ is a negation operator. On the other hand, it is not possible to 
define a negation operator for the tropical semiring. Hence, the syntax given in Def.~\ref{def:syntax} is proposed without considering the negation operator; otherwise, we can simplify it by removing $\cbot$ and $[\,]\phi$, since they can be respectively rewritten as $\neg \ctop$ and $\neg \langle\,\rangle\neg\phi$.

The semantics of a formula $\phi$ is given on a particular MLTS $M = \langle S, \mathit{Act}, \mathbb K, T, s_0 \rangle$, with the purpose to check the specification defined by $\phi$ over the behaviour of a weighted transition-system (in Sec.~\ref{sec:muGPA}, $M$ defines the behaviour of a GPA process). Note that, while in \cite{Andersen} the semantics of a formula computes the states $U \subseteq S$ that satisfy that formula, our semantics $\llbracket \, \rrbracket_M: (\Phi_M \times S) \longrightarrow K$  (see Tab.~\ref{tab:semantics}) computes a truth value (in $K$) for the same $U$. For instance, if we use the boolean semiring we always obtain $\ctop$ iff $U \not= \emptyset$, and $\cbot$ otherwise. It is not difficult to extend our semantics to also  return  $U$, as in \cite{Andersen}; however, in this work we are focused on computing a degree of satisfaction for $\phi$ (and $U$).

In Tab.~\ref{tab:semantics} and in the following (when clear from the context) we omit $M$ from $\llbracket \, \rrbracket_M$ for the sake of readability. The semantics is parametrised over a state $s \in S$, which is used to  consider only  the transitions that can be fired at a given step (labelled with an action $a$). The first $s$ will be the single initial state of the MLTS we define in Def.~\ref{def:MLTS}, \ie $s_0$.\footnote{Note that is also possible to let the semantics in Tab.~\ref{tab:semantics}  be parametrised on a set of states, by aggregating values on all the transitions originating from all of them. For instance, in case we have multiple initial states, as in~\cite{Buchholz}.}

\begin{table}[t]\hrulefill
\setlength{\tabcolsep}{0.6em}
\begin{center}

\begin{tabular}{ll}
$\llbracket k \rrbracket (s)$ &  $= k \in K \; \; \forall s \in S$\\   

$\llbracket \phi_1 + \phi_2\rrbracket (s)$& $= \llbracket \phi_1 \rrbracket(s) + \llbracket \phi_2\rrbracket(s)$ \\

$\llbracket \phi_1 \times \phi_2\rrbracket (s)$ & $= \llbracket \phi_1 \rrbracket(s) \times \llbracket \phi_2 \rrbracket (s)$ \\  

$\llbracket \phi_1 \sqcap \phi_2\rrbracket (s)$ & $= \llbracket \phi_1 \rrbracket(s) \sqcap \llbracket \phi_2 \rrbracket (s)$\\
&\\
$\llbracket \langle a\rangle \phi\rrbracket (s)$& $=\sum\limits_{R}( T(s,a,s') \times\llbracket \phi \rrbracket (s'))$ \\
  
&\\
$\llbracket  [a] \phi\rrbracket (s)$ & $ = 
\bigsqcap\limits_{R} (T(s,a,s') \times \llbracket \phi \rrbracket (s')) 
$\\
\\
\, & where $R = \{s' \in S \mid s\stackrel{a}{\rightarrow}s' \in T\}$
\end{tabular}

\end{center}\hrulefill 
\caption{Semantics of c-HM. $\sum\limits (\emptyset) =\cbot$ and $ \bigsqcap\limits (\emptyset) = \ctop$.} \label{tab:semantics}
\end{table}
%
\subsection{Interpreting c-HM over GPA}\label{sec:muGPA}
Both GPA and c-HM logic formulas can be interpreted on a MLTS. In this section, we focus on the interpretation of a c-HM formula $\phi$ on a GPA process $P$ to provide a notion of \emph{quantitative satisfiability} for the specification described by $\phi$, on the behaviour of a process $P$.
First of all, we define the projection of a process on an MLTS.
\begin{definition}[MLTS projection]
Given an MLTS $M = \langle S, \mathit{Act}, \mathbb K, T, s_0\rangle$, its projection over a process $P$ defined over the same $M$ is defined as $M \Downarrow_P = \langle S_P, \mathit{Act}, \mathbb K, T_P, s_0 \rangle$, where $S_P = \{ s \in S \mid s \in Der(P)\}$ and $T_P = \{(s,a,s') \in S\times Act \times S \mid s,s' \in S_{P} \wedge a \in Sort(P) \}$.\footnote{All the processes in parallel share the same $s_0$.}
\end{definition}

We are now ready to rephrase the notion of satisfiability to take into account a threshold $k$ ($k$-satisfiability):

\begin{definition}[$\models_k$]\label{def:sat}
A process $P$ satisfies  a c-HM formula $\phi$  with a threshold-value $t$, \ie $P \models_{t} \phi$, if and only if the interpretation of $\phi$ on $M \Downarrow_P$ is not worse than $t$. Formally:

$$
P \models_{t} \phi \Leftrightarrow  t \leq \llbracket \phi \rrbracket_{M\Downarrow_P}(s_0)
$$
\end{definition}

This means that $P$ is a model for a formula $\phi$ (with respect to a certain value $t$) iff the evaluation of  $\phi$ on  $P$ is not worse than $t$ in the partial order defined by $+$ in $\mathbb K$. It is worth noting that the interpretation of $\phi$ on $P$ is independent by the valuation of $P$ itself.


\emph{Remark 1.} Note that, if $P$ does not satisfy a formula $\phi$ then $\llbracket \phi \rrbracket_{M\Downarrow_P} = \cbot$. Consequently, the only $t$ such that $P \models_{t} \phi$ is $t= \cbot$. If $\llbracket \phi \rrbracket_{M\Downarrow_P} \not= \cbot$, then $\phi$ is satisfiable with a certain threshold $t \not= \cbot$.

\begin{example}\label{exe1}
In order to exemplify the concept expressed here, let us consider a formula $\phi$ stating that \emph{before opening a document \lq\lq file2" you have to close an already opened document \lq\lq file1" }. This is a security property aiming at preserving the confidentiality and integrity of the two documents. $\phi$ can be expressed by a c-HM formula as follows:
\small$$
\phi=[\texttt{open\_file1}]([\texttt{close\_file1}][\texttt{open\_file2}]\ctop \times [\texttt{open\_file2}]\cbot)
$$ 
The sub-formula after $\times$ (\ie $[\texttt{open\_file2}]$) is weighted with $\cbot$ because the opening of file2 has to be prevented in case file1 is not closed. Vice-versa, the left-side of $\times$ expresses the right behaviour, and thus it is weighted with $\ctop$. 

\noindent Then consider three different processes $P$ and $Q$, defined on  $\langle \mathbb{N}^+ \cup\{+\infty\}, min, \hat{+}, +\infty, 0 \rangle$ (\ie the tropical semiring):
\small$$
\begin{array}{lll}
P&=& (\texttt{open\_file1}, 5).(\texttt{close\_file1}, 4).0\\
Q &=& (\texttt{open\_file1}, 3).(\texttt{close\_file1}, 10).0\\
V &=& (\texttt{open\_file1}, 4).(\texttt{open\_file2}, 2).0
\end{array}
$$
According to our definition, $P \models_{11} \phi$ because, referring to Tab.~\ref{tab:semantics}, at the first step we consider the cost of the action $\texttt{open\_file1}$, \ie $5$, which is arithmetically  summed to \small$$\eval{([\texttt{close\_file1}][\texttt{open\_file2}]0 \: \hat{+} \: [\texttt{open\_file2}]\infty)}_{P'}$$
\noindent where $P'=(\texttt{close\_file1}, 4).0$. After $\texttt{close\_file1}$, the process halts, thus $\eval{[\texttt{open\_file2}]\infty} = 0$. Finally, we have $\eval{\phi}_{P}= 5 \: \hat{+} \: 4 \: \hat{+} \: 0 = 9$, which satisfies  the asked threshold $11$. $Q$ is evaluated in the same way, but since $\eval{\phi}_{Q}= 3 \: \hat{+} \: 10 \: \hat{+} \: 0 = 13$, we have that $P \not\models_{11} \phi$ because $11 \not\leq 14$. Therefore, even if there is a subset of $Q$ states that satisfies $\phi$, the degree satisfaction does not respect the requested threshold.
Finally,  $\phi$ is not satisfied by $V$ because $\eval{\phi}_{V}= 5 \: \hat{+} \:$ $\eval{([\texttt{close\_file1}][\texttt{open\_file2}]0 \: \hat{+} \: [\texttt{open\_file2}]$ $\infty)}_{V'} = 4 \: \hat{+} \: 2 \: \hat{+} \: \infty = \infty$.  
\end{example}

\section{Quantitative Partial Model Checking}\label{sec:qpmc}
\begin{table*}[!t] \noindent \hrulefill
{\small
$$
\begin{array}{lcl}
k_{\mathit{//}_{P}}  &=& k\\ 
(\phi_{1}\times \phi_{2})_{\mathit{//}_{P}} &=& (\phi_{1})_{\mathit{//}_{P}}\times (\phi_{2})_{\mathit{//}_{P}} \\
 
(\phi_{1}+ \phi_{2})_{\mathit{//}_{P}} &=& (\phi_{1})_{\mathit{//}_{P}} +  (\phi_{2})_{\mathit{//}_{P}}\\
 
 (\phi_{1}\sqcap \phi_{2})_{\mathit{//}_{P}} &=& (\phi_{1})_{\mathit{//}_{P}} \sqcap (\phi_{2})_{\mathit{//}_{P}} \\

([a]\phi_{1})_{\mathit{//}_{P}} &=& \left\{
\begin{array}{ll}
[a](\phi_{1})_{\mathit{//}_{P}} \sqcap
\bigsqcap\limits_{P\arco{a,k_a} P'} ( (k_a) \times (\phi_{1})_{\mathit{//}_{P'}}) &  a \not \in L\\
\bigsqcap\limits_{P\arco{a, k_a} P'} ( (k_a) \times [a] (\phi_{1})_{\mathit{//}_{P'}}) & a \in L
\end{array}
\right.
\\

(\langle a \rangle \phi_{1})_{\mathit{//}_{P}} &=& \left\{
\begin{array}{ll}
\langle a \rangle (\phi_{1})_{\mathit{//}_{P}} + \sum\limits_{P\arco{a,k_a} P'} ( (k_a) \times (\phi_{1})_{\mathit{//}_{P'}}) & a \not \in L\\
\sum\limits_{P\arco{a,k_a} P'} (( k_a)\times \langle a \rangle (\phi_{1})_{\mathit{//}_{P'}}) & a \in L
\end{array}
\right.
  
\end{array}
$$


}
\hrulefill
\caption{A QPMC function.}\label{def:qpmc}
\end{table*}


In this section we present a quantitative version of PMC~\cite{Andersen}, named QPMC, with respect to the parallel composition of GPA processes.
Such a function is defined in Tab.~\ref{def:qpmc}.
Being the logic closed, the interpretation of a formula obtained through the application of such function is straightforward. 
In Th.~\ref{main} we report a result similar (\ie weighted) to the one in \cite{Andersen}.

\begin{theorem}\label{main}
Given any two processes $P$ and $Q$ in parallel, and any c-HM formula $\phi$, then we have that
\small$$\llbracket \phi \rrbracket_{P \parallel_L Q} =  \llbracket \phi_{\mathit{//}_P} \rrbracket_Q.$$ 
\end{theorem}
\begin{proof}[Sketch]\footnote{The interested reader can find the complete proof in the appendix of the technical report \cite{QAPL15-TR}.}
The proposition is proved by induction on the complexity of a formula $\phi$. 
\begin{description}
\item[Base case, \underline{$\mathbf{\phi=k}$}:] According to Tab.~\ref{tab:semantics}, $\llbracket k \rrbracket_{P\Vert Q}= k=k_{\mathit{//}_{P}}=\llbracket k_{\mathit{//}_{P}} \rrbracket_{Q}$.
\item[Inductive Step:] As an example, let us now consider two different formulas:
\begin{description}
\item[\underline{$\mathbf{\phi=\phi_{1}\times \phi_{2}}$}:] According to Tab.~\ref{tab:semantics} we have that $\llbracket \phi \rrbracket_{P\Vert Q}=\llbracket \phi_{1} \times \phi_{2} \rrbracket_{P\Vert Q} = \llbracket  \phi_{1} \rrbracket_{P\Vert Q} \times \llbracket \phi_{2} \rrbracket_{P\Vert Q}$. By inductive hypothesis, $\llbracket  \phi_{1} \rrbracket_{P\Vert Q}= \llbracket  (\phi_{1})_{\mathit{//}_{P}} \rrbracket_{ Q}$ and $\llbracket  \phi_{2} \rrbracket_{P\Vert Q}=\llbracket  (\phi_{2})_{\mathit{//}_{P}} \rrbracket_{Q}$. Then $\llbracket  \phi_{1}  \rrbracket_{P\Vert Q} \times \llbracket  \phi_{2} \rrbracket_{P\Vert Q}= \llbracket  (\phi_{1})_{\mathit{//}_{P}} \rrbracket_{ Q} \times \llbracket  (\phi_{2})_{\mathit{//}_{P}} \rrbracket_{ Q} = \llbracket  (\phi_{1})_{\mathit{//}_{P}} \times (\phi_{2})_{\mathit{//}_{P}} \rrbracket_{Q} = \llbracket (\phi_{1} \times \phi_{2})_{\mathit{//}_{P}} \rrbracket_{Q}$.
\item The $+$ and the $\sqcap$ operators can be similarly proved.
\item[\underline{$\mathbf{\phi=\langle a \rangle \phi_{1}}$}:] According to Tab.~\ref{tab:semantics}, we have \small$$\llbracket \phi \rrbracket_{P\Vert Q}=\llbracket\langle a \rangle \phi_{1} \rrbracket_{P\Vert Q} = \sum\limits_{P\Vert Q \xrightarrow{(a, k_a)} (P\Vert Q)'}( (k_a) \times\llbracket \phi_{1} \rrbracket_{(P\Vert Q)'}).$$ Here we only prove  one of several possible cases: 
if $a \not \in L$, then \small$$\llbracket\langle a \rangle \phi_{1} \rrbracket_{P\Vert Q} = \sum\limits_{P\Vert Q \xrightarrow{(a, k_a)} (P\Vert Q)'}( (k_a) \times\llbracket \phi_{1} \rrbracket_{(P\Vert Q)'})$$ where $(P\Vert Q)'$ is equal to $P' \Vert Q$ if $P \xrightarrow{(a,k_{a})}P'$ or to $P\Vert Q'$ if $Q \xrightarrow{(a,k_{a})}Q'$. Hence, \small$$\sum\limits_{P\Vert Q \xrightarrow{(a, k_a)} (P\Vert Q)'}( (k_a) \times\llbracket \phi_{1} \rrbracket_{(P\Vert Q)'})= \sum\limits_{P \xrightarrow{(a, k_a)} P'}( (k_a) \times\llbracket \phi_{1} \rrbracket_{(P'\Vert Q)})+\sum\limits_{ Q \xrightarrow{(a, k_a)} Q'}( (k_a) \times\llbracket \phi_{1} \rrbracket_{(P\Vert Q')}).$$\\
 By inductive hypothesis, this is equal to \small$$\sum\limits_{P \xrightarrow{(a, k_a)} P'}( (k_a) \times\llbracket (\phi_{1})_{\mathit{//}_{P'}} \rrbracket_{Q})+\sum\limits_{ Q \xrightarrow{(a, k_a)} Q'}( (k_a) \times\llbracket (\phi_{1})_{\mathit{//}_{P}} \rrbracket_{Q'}).$$ \\Hence, $\llbracket\langle a \rangle \phi_{1} \rrbracket_{P\Vert Q}=\sum\limits_{P \xrightarrow{(a, k_a)} P'}( (k_a) \times\llbracket (\phi_{1})_{\mathit{//}_{P'}} \rrbracket_{Q}) + \llbracket\langle a \rangle (\phi_{1})_{\mathit{//}_{P}}\rrbracket_{Q} $.
On the other hand, \small$$\phi_{\mathit{//}_{P}} = (\langle a \rangle \phi_{1})_{\mathit{//}_{P}} = \langle a \rangle (\phi_{1})_{\mathit{//}_{P}} + \sum\limits_{P \xrightarrow{(a, k_a)} P'}( (k_a) \times (\phi_{1})_{\mathit{//}_{P'}})$$ and its semantics evaluation with respect to the process $Q$ is $\llbracket (\langle a \rangle \phi_{1})_{\mathit{//}_{P}}\rrbracket_{Q}=$ $$\llbracket \langle a \rangle (\phi_{1})_{\mathit{//}_{P}} + \sum\limits_{P \xrightarrow{(a, k_a)} P'}( (k_a) \times (\phi_{1})_{\mathit{//}_{P'}})\rrbracket _{Q}= \llbracket\langle a \rangle (\phi_{1})_{\mathit{//}_{P}}\rrbracket_{Q} + \llbracket \sum\limits_{P \xrightarrow{(a, k_a)} P'}( (k_a) \times (\phi_{1})_{\mathit{//}_{P'}})\rrbracket_{Q}.$$ Hence 
$\llbracket (\langle a \rangle \phi_{1})_{\mathit{//}_{P}}\rrbracket_{Q}= \llbracket\langle a \rangle (\phi_{1})_{\mathit{//}_{P}}\rrbracket_{Q} + \sum\limits_{P \xrightarrow{(a, k_a)} P'}( (k_a) \times\llbracket (\phi_{1})_{\mathit{//}_{P'}}\rrbracket_{Q})$. 
\end{description}
\end{description}
\end{proof}
\begin{example}\label{exe2}
Let us consider, the tropical semiring $\langle \mathbb{N}^+ \cup\{+\infty\}, min, \hat{+}, +\infty, 0 \rangle$, and two actions $\texttt{open}$ and $\texttt{close}$ ($L=\{\texttt{open}\}$). In addition, let us consider a formula $\phi= [\texttt{open}]\langle \texttt{close}\rangle\mathbf{1}$ stating that once a file is opened, then it has to be closed. We omit the name of the file because not significant for our example.
Let $P$ and $Q$ be two GPA processes: \\
\small$$
P = (\texttt{open}, 5).(\texttt{close}, 4).0 + (\texttt{open}, 6).0\quad \qquad
Q = (\texttt{open}, 4).(\texttt{close},3).0
$$
Let us consider the combined process $P \Vert_{L} Q$, where $P$ and $Q$ synchronise one another on actions in $L$, \ie on the action $\texttt{open}$. It is easy to see that $P \Vert_{L} Q \models_{20} \phi$.
By applying QPMC to $\phi$ w.r.t. $P$ we obtain:
\small$$
\begin{array}{lcl}
\phi_{//P} &=& (5 \times [\texttt{open}](\langle\texttt{close}\rangle\mathbf{1})_{//P'}) \sqcap (6 \times [\texttt{open}]\langle(\texttt{close}\rangle\mathbf{1})_{//P'})\\
&=& (5 \times [\texttt{open}](\langle\texttt{close}\rangle\mathbf{1}+ (4 \times 1))\sqcap (6 \times [\texttt{open}](\langle\texttt{close}\rangle\mathbf{1} + (4 \times 1))
\end{array}
$$\normalsize
\noindent where $+ = \min$, $\times \equiv \hat{+}$,  and $\sqcap  \equiv  \max$.
The QPMC function helps to understand which formula $Q$ has to satisfy in order to guarantee that the whole system satisfies the initial requirement. In this simple case, we know the behaviour of $Q$ and we can check if it quantitatively satisfies $\phi_{//P}$. To do this, we prove that $\llbracket \phi\rrbracket_{P\Vert_{L}Q} = \llbracket \phi_{\mathit{//P}}\rrbracket_{Q}$.
We have:
\small$$
\begin{array}{lcl}

\llbracket \phi\rrbracket_{P\Vert_{L}Q} &=& max(9+(min (4,3)), 10 +3)= max( 12,13) = 13,\\
\llbracket \phi_{\mathit{//P}}\rrbracket_{Q} &=& max (5 +(4 + min (3,5)), 6 +(4 + min (3,5)))= max(12, 13) = 13.
\end{array}
$$
%
\end{example}


%

%% file: 3-Related.tex
\section{Related Work}\label{sec:related}
The aim of this work is to present a semiring-based formal framework where to deal with quantitative specification of security in combined systems.
We dedicate the first part of this section to alternative definitions of quantitative bisimulation relations, in some cases even not applied to security (e.g., \cite{MiculanP13}).

In \cite{MiculanP13} the authors  extend \emph{Weighted Labelled Transition Systems} (WLTS) towards other behavioural equivalences, by considering semirings of weights. The main result of such work is the definition of a general notion of \emph{weak weighted bisimulation}. 
They show that this relation coincides with the usual weak bisimulation in case of non-deterministic and fully-probabilistic systems. Moreover, it can also be extended towards kinds of LTSs where this notion is currently missing (e.g., stochastic systems). In Def.~\ref{def:epsilon-trace} we also relax quantitative weak-bisimulation to weak $\epsilon$-bisimulation.

In \cite{aldini} the authors address the problem of providing a quantitative estimation of the confidentiality of a system by measuring its information leakage. In our analysis the most powerful adversary is measured via a notion of approximate process equivalence. In practice, the lack of information leakage is expressed by a successful weak probabilistic bisimulation based check. Whenever such a check fails, approximate relations relax the conditions imposed by the weak probabilistic-bisimulation, in such a way that the level of approximation represents an estimate of the amount of information leakage. Our notion of $\epsilon$-bisimulation is very close to \cite{aldini}, except that we generalise it by using semiring operators.

Even the approach in \cite{continuous} bounds the distance between the transitions of two states: if their distance is less equal than a threshold $\delta$, and this holds for all the states of two processes $P_1$ and $P_2$, such processes are said to be approximately bisimilar with a $\delta$-precision. The motivations is that, interacting with the physical world, exact relationships are restrictive and not robust.

The literature also proposes works using fuzzy weights (in this work we have the fuzzy semiring): in \cite{cao} a notion of behavioural distance is given to measure the behavioural similarity of non-deterministic fuzzy-transition systems. Two systems are at zero distance if and only if they are bisimilar.

Considering the second fragment of the paper, no direct comparison is available for QPMC. Nevertheless, our c-semiring H-M Logic (see Sec.~\ref{sec:logic}) has been inspired by the work in \cite{LlM05}. 
Some examples of quantitative temporal logic are \cite{stoelinga,borto1}. In \cite{stoelinga} the authors present \emph{QLTL}, a quantitative analogue of \emph{LTL} and presents algorithms for Model Checking it over a quantitative version of Kripke structures and Markov chains. Thus, weights are in the interval of Real numbers $[0,1]$. In \cite{borto1} the authors  combine robustness scores with the satisfaction probability to optimise some control parameters of a stochastic model: the goal is to best maximise robustness of the desired specifications. However, even this approach is focused on Continuous-Time Markov Chains, and not on semiring algebraic-structures.

Non-binary measures of security have been considered for access control systems by Cheng et al.~\cite{Cheng:2007:FMS:1263552.1264209}. The level of security should correspond to a fuzzy domain rather than a strict separation between what is secure and what is not.  Zhang et al. define with the BARAC model~\cite{zhang06Policy} a notion of benefit for each access, with the underlying idea that allowing an access comes with a benefit for the system. The ``value'' of an access or an action can be for instance calculated using market-based techniques~\cite{MCR08}.

From a different perspective, 
Bielova and Massacci propose in~\cite{DBLP:conf/essos/BielovaM11} a notion of distance among traces, thus expressing that if a trace is not secure, it should be edited to a secure trace close to the non-secure one, thus characterising enforcement strategies by the distance from the original trace they create. 
In~\cite{DrabikMM12}, a similar notion of cost has been introduced
following some intuitive leads given in~\cite{Martinelli:2012:QQE:2404707.2404711} in order to move from qualitative to quantitative enforcement.
Semirings have been used by Bistarelli et al. in the context of access control~\cite{DBLP:journals/cma/BistarelliMS12} and trust systems~\cite{DBLP:journals/scn/BistarelliFOS10}. Here we use them in the context of enforcement mechanism defined trough process algebra, following the approach by Buchholz and Kemper~\cite{Buchholz}.

%% file: 5-Conclusion.tex
\section{Conclusion}\label{sec:conclusion}
We have introduced two different formal-frameworks oriented to the specification of quantitative properties on a GPA-process. Both of the frameworks are have a common \emph{trait d'union} consisting in the use of c-semiring structures to represent transition costs. By taking advantage of such costs, we can constrain classical qualitative-relations between two processes, as we do as our first contribute for trace equivalence and weak bisimulation equivalence. In practice we parametrise the weak bisimulation notion given in \cite{aldini} by allowing for different metrics, and not probability scores only. At the same time we refine the definition of semiring-based bisimulation given in \cite{MiculanP13}, by extending the relation in order to consider $\epsilon$-close 
processes. 
As a second result, we propose
a way to express security constraints via a quantitative version of the Hennessy-Milner logic, and  a method for specifying the security of a system through a quantitative version of PMC, which allows us to move a process from the parallel computation to a formula $\phi$. If the system satisfies a security property with a value $k$ worse than  $t$ (a \emph{security threshold}), then the system is not \emph{quantitatively secure}.
In this way we can use this threshold to tradeoff security and functionality/performance requirements.

The essence of the paper is to advance the same basic bricks  (\ie GPA and semirings) with the purpose to enhance two different quantitative frameworks (\ie process equivalences and PMC), which are  nevertheless related by the common purpose of (security) property specification. Of course both of the frameworks can be independently (but still interlacedly) developed to offer a complete specification and validation tool on their own, as the following ideas on future work suggest.

In the future we aim to extend both the approaches in different directions. As an ongoing work, we are investigating on the definition of the characteristic formula of a processes, with respect to each bisimulation equivalence definitions we have provided in Sec.~\ref{sec:qgndc}. In such way, we will be able to compare the effectiveness of the two proposed approaches.
Furthermore, we aim to extend both of them in order to not only use them for the specification has but also for the analysis. Indeed, referring to the former approach, we need to investigate on the characterisation of the most powerful attacker in order to compare the system under attack, with respect to the expected behaviour. This can be done only under certain constraints on the considered equivalences.
Referring on the latter approach, we need to elaborate a satisfiability procedure for the quantitative logic we have introduced here in order to verify if the system under investigation is secure or not, \ie it satisfies the security requirement.

Another possible direction we would like to investigate is the identification of comparative strategies based on the (partial or total) ordering of the semiring. In this way we can compare different strategies and finally synthesise the best one (whether it exists).
Another direction is the extension of the framework to use more than one measure associated to each action in order to evaluate a process. Such measures can be combined and ordered, \eg by using the lexicographical ordering, in such a way that controlling strategies can be selected with respect to the optimisation of the trade-off between some of them.


%

%% file: main.bbl
\begin{thebibliography}{10}
\providecommand{\bibitemdeclare}[2]{}
\providecommand{\surnamestart}{}
\providecommand{\surnameend}{}
\providecommand{\urlprefix}{Available at }
\providecommand{\url}[1]{\texttt{#1}}
\providecommand{\href}[2]{\texttt{#2}}
\providecommand{\urlalt}[2]{\href{#1}{#2}}
\providecommand{\doi}[1]{doi:\urlalt{http://dx.doi.org/#1}{#1}}
\providecommand{\bibinfo}[2]{#2}

\bibitemdeclare{article}{aldini}
\bibitem{aldini}
\bibinfo{author}{A.~\surnamestart Aldini\surnameend} \& \bibinfo{author}{A.~Di
  \surnamestart Pierro\surnameend} (\bibinfo{year}{2008}):
  \emph{\bibinfo{title}{Estimating the maximum information leakage}}.
\newblock {\sl \bibinfo{journal}{Int. J. Inf. Sec.}}
  \bibinfo{volume}{7}(\bibinfo{number}{3}), pp. \bibinfo{pages}{219--242},
  \doi{10.1007/s10207-007-0050-x}.

\bibitemdeclare{inproceedings}{Andersen}
\bibitem{Andersen}
\bibinfo{author}{H.~R. \surnamestart Andersen\surnameend}
  (\bibinfo{year}{1995}): \emph{\bibinfo{title}{Partial Model Checking}}.
\newblock In: {\sl \bibinfo{booktitle}{LICS '95}}, \bibinfo{publisher}{IEEE
  Computer Society}, p. \bibinfo{pages}{398}, \doi{10.1109/LICS.1995.523274}.

\bibitemdeclare{inproceedings}{borto1}
\bibitem{borto1}
\bibinfo{author}{E.~\surnamestart Bartocci\surnameend},
  \bibinfo{author}{L.~\surnamestart Bortolussi\surnameend},
  \bibinfo{author}{L.~\surnamestart Nenzi\surnameend} \&
  \bibinfo{author}{G.~\surnamestart Sanguinetti\surnameend}
  (\bibinfo{year}{2013}): \emph{\bibinfo{title}{On the Robustness of Temporal
  Properties for Stochastic Models}}.
\newblock In: {\sl \bibinfo{booktitle}{2nd International Workshop on Hybrid
  Systems and Biology}}, {\sl \bibinfo{series}{{EPTCS}}} \bibinfo{volume}{125},
  pp. \bibinfo{pages}{3--19}, \doi{10.4204/EPTCS.125.1}.

\bibitemdeclare{inproceedings}{DBLP:conf/essos/BielovaM11}
\bibitem{DBLP:conf/essos/BielovaM11}
\bibinfo{author}{N.~\surnamestart Bielova\surnameend} \&
  \bibinfo{author}{F.~\surnamestart Massacci\surnameend}
  (\bibinfo{year}{2011}): \emph{\bibinfo{title}{Predictability of
  Enforcement}}.
\newblock In: {\sl \bibinfo{booktitle}{Proceedings of ESSoS 2011}},
  \bibinfo{volume}{6542}, \bibinfo{publisher}{Springer}, pp.
  \bibinfo{pages}{73--86}, \doi{10.1007/978-3-642-19125-1\_6}.

\bibitemdeclare{article}{DBLP:journals/scn/BistarelliFOS10}
\bibitem{DBLP:journals/scn/BistarelliFOS10}
\bibinfo{author}{S.~\surnamestart Bistarelli\surnameend},
  \bibinfo{author}{S.~N. \surnamestart Foley\surnameend},
  \bibinfo{author}{B.~\surnamestart O'Sullivan\surnameend} \&
  \bibinfo{author}{F.~\surnamestart Santini\surnameend} (\bibinfo{year}{2010}):
  \emph{\bibinfo{title}{Semiring-based frameworks for trust propagation in
  small-world networks and coalition formation criteria}}.
\newblock {\sl \bibinfo{journal}{Security and Communication Networks}}
  \bibinfo{volume}{3}(\bibinfo{number}{6}), pp. \bibinfo{pages}{595--610},
  \doi{10.1002/sec.252}.

\bibitemdeclare{article}{DBLP:journals/cma/BistarelliMS12}
\bibitem{DBLP:journals/cma/BistarelliMS12}
\bibinfo{author}{S.~\surnamestart Bistarelli\surnameend},
  \bibinfo{author}{F.~\surnamestart Martinelli\surnameend} \&
  \bibinfo{author}{F.~\surnamestart Santini\surnameend} (\bibinfo{year}{2012}):
  \emph{\bibinfo{title}{A semiring-based framework for the deduction/abduction
  reasoning in access control with weighted credentials}}.
\newblock {\sl \bibinfo{journal}{CAMWA}}
  \bibinfo{volume}{64}(\bibinfo{number}{4}), pp. \bibinfo{pages}{447--462},
  \doi{10.1016/j.camwa.2011.12.017}.

\bibitemdeclare{article}{jacm97}
\bibitem{jacm97}
\bibinfo{author}{S.~\surnamestart Bistarelli\surnameend},
  \bibinfo{author}{U.~\surnamestart Montanari\surnameend} \&
  \bibinfo{author}{F.~\surnamestart Rossi\surnameend} (\bibinfo{year}{1997}):
  \emph{\bibinfo{title}{Semiring-based constraint satisfaction and
  optimization}}.
\newblock {\sl \bibinfo{journal}{J. ACM}}
  \bibinfo{volume}{44}(\bibinfo{number}{2}), pp. \bibinfo{pages}{201--236},
  \doi{10.1145/256303.256306}.

\bibitemdeclare{inproceedings}{gadducci06}
\bibitem{gadducci06}
\bibinfo{author}{Stefano \surnamestart Bistarelli\surnameend} \&
  \bibinfo{author}{Fabio \surnamestart Gadducci\surnameend}
  (\bibinfo{year}{2006}): \emph{\bibinfo{title}{Enhancing Constraints
  Manipulation in Semiring-based Formalisms}}.
\newblock In: {\sl \bibinfo{booktitle}{Proceedings of the 2006 Conference on
  ECAI 2006: 17th European Conference on Artificial Intelligence}},
  \bibinfo{publisher}{IOS Press}, pp. \bibinfo{pages}{63--67}.
\newblock \urlprefix\url{http://dl.acm.org/citation.cfm?id=1567016.1567036}.

\bibitemdeclare{book}{residuation}
\bibitem{residuation}
\bibinfo{author}{T.~S. \surnamestart Blyth\surnameend} \&
  \bibinfo{author}{M.~F. \surnamestart Janowitz\surnameend}
  (\bibinfo{year}{1972}): \emph{\bibinfo{title}{Residuation theory}}.
\newblock \bibinfo{volume}{102}, \bibinfo{publisher}{Pergamon press Oxford}.

\bibitemdeclare{inproceedings}{Buchholz}
\bibitem{Buchholz}
\bibinfo{author}{Peter \surnamestart Buchholz\surnameend} \&
  \bibinfo{author}{Peter \surnamestart Kemper\surnameend}
  (\bibinfo{year}{2001}): \emph{\bibinfo{title}{Quantifying the Dynamic
  Behavior of Process Algebras}}.
\newblock In: {\sl \bibinfo{booktitle}{Process Algebra and Probabilistic
  Methods, Performance Modeling and Verification: Joint International Workshop,
  {PAPM-PROBMIV} 2001, Aachen, Germany, September 12-14, 2001, Proceedings}},
  pp. \bibinfo{pages}{184--199}, \doi{10.1007/3-540-44804-7\_12}.

\bibitemdeclare{article}{cao}
\bibitem{cao}
\bibinfo{author}{Y.~\surnamestart Cao\surnameend}, \bibinfo{author}{S.~X.
  \surnamestart Sun\surnameend}, \bibinfo{author}{H.~\surnamestart
  Wang\surnameend} \& \bibinfo{author}{G.~\surnamestart Chen\surnameend}
  (\bibinfo{year}{2013}): \emph{\bibinfo{title}{A Behavioral Distance for
  Fuzzy-Transition Systems}}.
\newblock {\sl \bibinfo{journal}{{IEEE} T. Fuzzy Systems}}
  \bibinfo{volume}{21}(\bibinfo{number}{4}), pp. \bibinfo{pages}{735--747},
  \doi{10.1109/TFUZZ.2012.2230177}.

\bibitemdeclare{inproceedings}{Cheng:2007:FMS:1263552.1264209}
\bibitem{Cheng:2007:FMS:1263552.1264209}
\bibinfo{author}{Pau{-}Chen \surnamestart Cheng\surnameend},
  \bibinfo{author}{Pankaj \surnamestart Rohatgi\surnameend},
  \bibinfo{author}{Claudia \surnamestart Keser\surnameend},
  \bibinfo{author}{Paul~A. \surnamestart Karger\surnameend},
  \bibinfo{author}{Grant~M. \surnamestart Wagner\surnameend} \&
  \bibinfo{author}{Angela~Schuett \surnamestart Reninger\surnameend}
  (\bibinfo{year}{2007}): \emph{\bibinfo{title}{Fuzzy Multi-Level Security: An
  Experiment on Quantified Risk-Adaptive Access Control}}.
\newblock In: {\sl \bibinfo{booktitle}{2007 {IEEE} Symposium on Security and
  Privacy (S{\&}P 2007), 20-23 May 2007, Oakland, California, {USA}}}, pp.
  \bibinfo{pages}{222--230}, \doi{10.1109/SP.2007.21}.

\bibitemdeclare{inproceedings}{DrabikMM12}
\bibitem{DrabikMM12}
\bibinfo{author}{Peter \surnamestart Dr{\'{a}}bik\surnameend},
  \bibinfo{author}{Fabio \surnamestart Martinelli\surnameend} \&
  \bibinfo{author}{Charles \surnamestart Morisset\surnameend}
  (\bibinfo{year}{2012}): \emph{\bibinfo{title}{Cost-Aware Runtime Enforcement
  of Security Policies}}.
\newblock In: {\sl \bibinfo{booktitle}{Security and Trust Management - 8th
  International Workshop, {STM} 2012, Pisa, Italy, September 13-14, 2012,
  Revised Selected Papers}}, pp. \bibinfo{pages}{1--16},
  \doi{10.1007/978-3-642-38004-4\_1}.

\bibitemdeclare{article}{stoelinga}
\bibitem{stoelinga}
\bibinfo{author}{Marco \surnamestart Faella\surnameend}, \bibinfo{author}{Axel
  \surnamestart Legay\surnameend} \& \bibinfo{author}{Mari{\"{e}}lle
  \surnamestart Stoelinga\surnameend} (\bibinfo{year}{2008}):
  \emph{\bibinfo{title}{Model Checking Quantitative Linear Time Logic}}.
\newblock {\sl \bibinfo{journal}{Electr. Notes Theor. Comput. Sci.}}
  \bibinfo{volume}{220}(\bibinfo{number}{3}), pp. \bibinfo{pages}{61--77},
  \doi{10.1016/j.entcs.2008.11.019}.

\bibitemdeclare{inproceedings}{683116}
\bibitem{683116}
\bibinfo{author}{Riccardo \surnamestart Focardi\surnameend} \&
  \bibinfo{author}{Roberto \surnamestart Gorrieri\surnameend}
  (\bibinfo{year}{2000}): \emph{\bibinfo{title}{Classification of Security
  Properties (Part {I:} Information Flow)}}.
\newblock In: {\sl \bibinfo{booktitle}{Foundations of Security Analysis and
  Design, Tutorial Lectures [revised versions of lectures given during the
  {IFIP} {WG} 1.7 International School on Foundations of Security Analysis and
  Design, {FOSAD} 2000, Bertinoro, Italy, September 2000]}}, pp.
  \bibinfo{pages}{331--396}, \doi{10.1007/3-540-45608-2\_6}.

\bibitemdeclare{inproceedings}{DBLP:conf/fosad/FocardiGM02}
\bibitem{DBLP:conf/fosad/FocardiGM02}
\bibinfo{author}{Riccardo \surnamestart Focardi\surnameend},
  \bibinfo{author}{Roberto \surnamestart Gorrieri\surnameend} \&
  \bibinfo{author}{Fabio \surnamestart Martinelli\surnameend}
  (\bibinfo{year}{2002}): \emph{\bibinfo{title}{Classification of Security
  Properties - Part {II:} Network Security}}.
\newblock In: {\sl \bibinfo{booktitle}{Foundations of Security Analysis and
  Design II, {FOSAD} 2001/2002 Tutorial Lectures}}, pp.
  \bibinfo{pages}{139--185}, \doi{10.1007/978-3-540-24631-2\_4}.

\bibitemdeclare{inproceedings}{FocardiM99}
\bibitem{FocardiM99}
\bibinfo{author}{Riccardo \surnamestart Focardi\surnameend} \&
  \bibinfo{author}{Fabio \surnamestart Martinelli\surnameend}
  (\bibinfo{year}{1999}): \emph{\bibinfo{title}{A Uniform Approach for the
  Definition of Security Properties}}.
\newblock In: {\sl \bibinfo{booktitle}{FM'99 - Formal Methods, World Congress
  on Formal Methods in the Development of Computing Systems, Toulouse, France,
  September 20-24, 1999, Proceedings, Volume {I}}}, pp.
  \bibinfo{pages}{794--813}, \doi{10.1007/3-540-48119-2\_44}.

\bibitemdeclare{article}{continuous}
\bibitem{continuous}
\bibinfo{author}{A.~\surnamestart Girard\surnameend} \& \bibinfo{author}{G.~J.
  \surnamestart Pappas\surnameend} (\bibinfo{year}{2007}):
  \emph{\bibinfo{title}{Approximation Metrics for Discrete and Continuous
  Systems}}.
\newblock {\sl \bibinfo{journal}{{IEEE} Trans. Automat. Contr.}}
  \bibinfo{volume}{52}(\bibinfo{number}{5}), pp. \bibinfo{pages}{782--798},
  \doi{10.1109/TAC.2007.895849}.

\bibitemdeclare{inproceedings}{NI}
\bibitem{NI}
\bibinfo{author}{J.~A. \surnamestart Goguen\surnameend} \&
  \bibinfo{author}{J.~\surnamestart Meseguer\surnameend}
  (\bibinfo{year}{1982}): \emph{\bibinfo{title}{Security Policy and Security
  Models}}.
\newblock In: {\sl \bibinfo{booktitle}{Proc. of the 1982 Symposium on Security
  and Privacy}}, \bibinfo{publisher}{IEEE Press}, pp. \bibinfo{pages}{11--20},
  \doi{10.1109/SP.1982.10014}.

\bibitemdeclare{book}{golan}
\bibitem{golan}
\bibinfo{author}{J.~\surnamestart Golan\surnameend} (\bibinfo{year}{2003}):
  \emph{\bibinfo{title}{Semirings and affine equations over them: theory and
  applications}}.
\newblock \bibinfo{publisher}{Kluwer Academic Pub.},
  \doi{10.1007/978-94-017-0383-3}.

\bibitemdeclare{article}{dagstuhl}
\bibitem{dagstuhl}
\bibinfo{author}{Boris \surnamestart K{\"{o}}pf\surnameend},
  \bibinfo{author}{Pasquale \surnamestart Malacaria\surnameend} \&
  \bibinfo{author}{Catuscia \surnamestart Palamidessi\surnameend}
  (\bibinfo{year}{2012}): \emph{\bibinfo{title}{Quantitative Security Analysis
  (Dagstuhl Seminar 12481)}}.
\newblock {\sl \bibinfo{journal}{Dagstuhl Reports}}
  \bibinfo{volume}{2}(\bibinfo{number}{11}), pp. \bibinfo{pages}{135--154},
  \doi{10.4230/DagRep.2.11.135}.

\bibitemdeclare{inproceedings}{GLMM09}
\bibitem{GLMM09}
\bibinfo{author}{Gabriele \surnamestart Lenzini\surnameend},
  \bibinfo{author}{Fabio \surnamestart Martinelli\surnameend},
  \bibinfo{author}{Ilaria \surnamestart Matteucci\surnameend} \&
  \bibinfo{author}{Stefania \surnamestart Gnesi\surnameend}
  (\bibinfo{year}{2008}): \emph{\bibinfo{title}{A Uniform Approach to Security
  and Fault-Tolerance Specification and Analysis}}.
\newblock In: {\sl \bibinfo{booktitle}{Architecting Dependable Systems {VI}}},
  pp. \bibinfo{pages}{172--201}, \doi{10.1007/978-3-642-10248-6\_8}.

\bibitemdeclare{article}{LlM05}
\bibitem{LlM05}
\bibinfo{author}{A.~\surnamestart Lluch-Lafuente\surnameend} \&
  \bibinfo{author}{U.~\surnamestart Montanari\surnameend}
  (\bibinfo{year}{2005}): \emph{\bibinfo{title}{Quantitative mu-calculus and
  {CTL} defined over constraint semirings}}.
\newblock {\sl \bibinfo{journal}{TCS}}
  \bibinfo{volume}{346}(\bibinfo{number}{1}), pp. \bibinfo{pages}{135--160},
  \doi{10.1016/j.tcs.2005.08.006}.

\bibitemdeclare{article}{DBLP:journals/entcs/MartinelliM07}
\bibitem{DBLP:journals/entcs/MartinelliM07}
\bibinfo{author}{Fabio \surnamestart Martinelli\surnameend} \&
  \bibinfo{author}{Ilaria \surnamestart Matteucci\surnameend}
  (\bibinfo{year}{2007}): \emph{\bibinfo{title}{An Approach for the
  Specification, Verification and Synthesis of Secure Systems}}.
\newblock {\sl \bibinfo{journal}{Electr. Notes Theor. Comput. Sci.}}
  \bibinfo{volume}{168}, pp. \bibinfo{pages}{29--43},
  \doi{10.1016/j.entcs.2006.12.003}.

\bibitemdeclare{inproceedings}{Martinelli:2012:QQE:2404707.2404711}
\bibitem{Martinelli:2012:QQE:2404707.2404711}
\bibinfo{author}{Fabio \surnamestart Martinelli\surnameend},
  \bibinfo{author}{Ilaria \surnamestart Matteucci\surnameend} \&
  \bibinfo{author}{Charles \surnamestart Morisset\surnameend}
  (\bibinfo{year}{2012}): \emph{\bibinfo{title}{From Qualitative to
  Quantitative Enforcement of Security Policy}}.
\newblock In: {\sl \bibinfo{booktitle}{Computer Network Security - 6th
  International Conference on Mathematical Methods, Models and Architectures
  for Computer Network Security, {MMM-ACNS} 2012, St. Petersburg, Russia,
  October 17-19, 2012. Proceedings}}, pp. \bibinfo{pages}{22--35},
  \doi{10.1007/978-3-642-33704-8\_3}.

\bibitemdeclare{misc}{QAPL15-TR}
\bibitem{QAPL15-TR}
\bibinfo{author}{Fabio \surnamestart Martinelli\surnameend},
  \bibinfo{author}{Ilaria \surnamestart Matteucci\surnameend} \&
  \bibinfo{author}{Francesco \surnamestart Santini\surnameend}
  (\bibinfo{year}{2015}): \emph{\bibinfo{title}{{Semiring-based Specification
  Approaches for Quantitative Security}}}.
\newblock \bibinfo{note}{Technical Report TR-IIT-08-2015}.

\bibitemdeclare{article}{MiculanP13}
\bibitem{MiculanP13}
\bibinfo{author}{Marino \surnamestart Miculan\surnameend} \&
  \bibinfo{author}{Marco \surnamestart Peressotti\surnameend}
  (\bibinfo{year}{2013}): \emph{\bibinfo{title}{Weak bisimulations for labelled
  transition systems weighted over semirings}}.
\newblock {\sl \bibinfo{journal}{CoRR}} \bibinfo{volume}{abs/1310.4106}.
\newblock \urlprefix\url{http://arxiv.org/abs/1310.4106}.

\bibitemdeclare{book}{milner}
\bibitem{milner}
\bibinfo{author}{R.~\surnamestart Milner\surnameend} (\bibinfo{year}{1999}):
  \emph{\bibinfo{title}{Communicating and mobile systems: the $\pi$-calculus}}.
\newblock \bibinfo{publisher}{Cambridge University Press}.

\bibitemdeclare{inproceedings}{MCR08}
\bibitem{MCR08}
\bibinfo{author}{Ian \surnamestart Molloy\surnameend},
  \bibinfo{author}{Pau{-}Chen \surnamestart Cheng\surnameend} \&
  \bibinfo{author}{Pankaj \surnamestart Rohatgi\surnameend}
  (\bibinfo{year}{2008}): \emph{\bibinfo{title}{Trading in risk: using markets
  to improve access control}}.
\newblock In: {\sl \bibinfo{booktitle}{Proceedings of the 2008 Workshop on New
  Security Paradigms, Lake Tahoe, CA, USA, September 22-25, 2008}}, pp.
  \bibinfo{pages}{107--125}, \doi{10.1145/1595676.1595694}.

\bibitemdeclare{article}{csop}
\bibitem{csop}
\bibinfo{author}{Sergiu \surnamestart Rudeanu\surnameend} \&
  \bibinfo{author}{Dragos \surnamestart Vaida\surnameend}
  (\bibinfo{year}{2004}): \emph{\bibinfo{title}{Semirings in Operations
  Research and Computer Science: More Algebra}}.
\newblock {\sl \bibinfo{journal}{Fundam. Inform.}}
  \bibinfo{volume}{61}(\bibinfo{number}{1}), pp. \bibinfo{pages}{61--85}.
\newblock
  \urlprefix\url{http://content.iospress.com/articles/fundamenta-informaticae/fi61-1-06}.

\bibitemdeclare{inproceedings}{zhang06Policy}
\bibitem{zhang06Policy}
\bibinfo{author}{Lei \surnamestart Zhang\surnameend},
  \bibinfo{author}{Alexander \surnamestart Brodsky\surnameend} \&
  \bibinfo{author}{Sushil \surnamestart Jajodia\surnameend}
  (\bibinfo{year}{2006}): \emph{\bibinfo{title}{Toward Information Sharing:
  Benefit And Risk Access Control {(BARAC)}}}.
\newblock In: {\sl \bibinfo{booktitle}{7th {IEEE} International Workshop on
  Policies for Distributed Systems and Networks {(POLICY} 2006), 5-7 June 2006,
  London, Ontario, Canada}}, pp. \bibinfo{pages}{45--53},
  \doi{10.1109/POLICY.2006.36}.

\end{thebibliography}
